\documentclass[11pt]{article}
\usepackage[utf8]{inputenc}
\usepackage[margin=1.3in]{geometry}
\pdfoutput=1
\usepackage[T1]{fontenc}
\usepackage{kpfonts,baskervald}

\usepackage{cite}

\usepackage{amsmath}
\usepackage{stmaryrd}

\usepackage[table,xcdraw]{xcolor}
\usepackage{enumitem}

\usepackage[all]{xy}

\usepackage{adjustbox}
\usepackage{tikz}
\usetikzlibrary{positioning}
\usetikzlibrary{cd}
\usetikzlibrary{decorations.markings}
\usepackage{pgfplots}
\usepackage[labelfont=bf]{caption}
\usepackage{blindtext}
\pgfplotsset{compat=1.17}

\usepackage{amsthm}
\theoremstyle{definition}
\newtheorem{defn}{Definition}[section]

\newtheorem{theorem}{Theorem}[section]
\newtheorem{proposition}{Proposition}[section]

\numberwithin{equation}{section}

\usepackage{hyperref}
\definecolor{airforceblue}{rgb}{0.36, 0.54, 0.66}
\hypersetup{
    colorlinks,
    citecolor=purple,
    filecolor=black,
    linkcolor=airforceblue,
    urlcolor=blue
}
\usepackage[nameinlink]{cleveref}

\renewcommand*\rm[1]{\mathrm{#1}}

\newcommand\HH{\mathrm{H}}

\newcommand*\I{\mathrm{i}}

\newcommand*\cM{\mathcal{M}}
\newcommand*\cN{\mathcal{N}}

\newcommand*\cP{\mathcal{P}}
\newcommand*\cQ{\mathcal{Q}}
\newcommand*\cR{\mathcal{R}}

\newcommand*\cT{\mathcal{T}}

\newcommand*\cZ{\mathcal{Z}}

\newcommand*\bbC{\mathbb{C}}

\newcommand*\bbF{\mathbb{F}}

\newcommand*\bbN{\mathbb{N}}

\newcommand*\bbZ{\mathbb{Z}}

\newcommand{\ab}[1]{#1^{\mathrm{ab}}}
\newcommand{\GL}[1]{\mathrm{GL}(#1)}
\newcommand{\PGL}[1]{\mathrm{PGL}(#1)}
\newcommand*\U{\mathrm{U}}
\newcommand*\SU{\mathrm{SU}}
\newcommand*\Alt{\mathrm{Alt}}
\newcommand*\Sym{\mathrm{Sym}}
\newcommand*\bbCt{\mathbb{C}^\times}
\newcommand*\QFT{\mathbf{QFT}}

\title{\vspace{1cm}\bf{McKay quivers and decomposition}\vspace{0.2cm}}
\date{\vspace{0.2cm} \today}
\author{Shani Meynet$^\circ$ and Robert Moscrop$^\bullet$\\
Department of Mathematics, Uppsala University\\\
Box 480, SE-75106 Uppsala, Sweden}

\begin{document}

\begin{titlepage}
\maketitle
\thispagestyle{empty}
\vspace{0.2cm}
\begin{abstract}
	\noindent When a quantum field theory in $d$-spacetime dimensions possesses a global $(d-1)$-form symmetry, it can decompose into disjoint unions of other theories. This is reflected in the physical quantities of the theory and can be used to study properties of the constituent theories. In this note we highlight the equivalence between the decomposition of orbifold $\sigma$-models and disconnected McKay quivers. 
	Specifically, we show in numerous examples that each component of a McKay quiver can be given definitive geometric meaning through the decomposition formulae. In addition, we give a purely group and representation theoretic derivation of the quivers for the cases where the trivially acting part of the orbifold group is central. As expected, the resulting quivers are compatible with the case of $\sigma$-models on `banded' gerbes. 
\end{abstract}
\vfill

\begin{flushright}
\noindent $\underline{\quad\quad\quad\quad\quad\quad\quad\quad\quad\quad\quad\;\;}$\\
$^\circ$\footnotesize{{\tt \href{mailto:shani.meynet@math.uu.se}{shani.meynet@math.uu.se}}}\\
$^\bullet$\footnotesize{{\tt \href{mailto:robert.moscrop@math.uu.se}{robert.moscrop@math.uu.se}}}
\end{flushright}
\end{titlepage}


\setcounter{page}{1}
{\hypersetup{linkcolor=black}
\tableofcontents
}

\section{Introduction}

Decomposition is a phenomenon observed in $d$-dimensional quantum field theories (QFTs) with a $(d-1)$-form symmetry. In these cases, the QFT can be seen as equivalent to the disjoint union of other theories often called `universes' \cite{Hellerman:2006zs,Pantev:2005rh, Robbins:2020msp,Pantev:2022pbf,Sharpe:2021srf,Robbins:2021ibx,Sharpe:2014tca, Komargodski:2020mxz}. This striking property was first observed in the context of string compactification on Calabi-Yau stacks\footnote{See \cite{Sharpe:2001bs, Sharpe:2001zp,Pantev:2005wj} for early work on stacks in string theory.} and gerbes \cite{Hellerman:2006zs}. Since then, such behaviour has been observed in  many different contexts \cite{Sharpe:2022ene}:
\begin{itemize}
    \item Two-dimensional $\U(1)$ gauge theory with-non minimal charges is equivalent to the union of independent $\U(1)$ theories with minimal charges \cite{Hellerman:2006zs}.
    
    \item A two-dimensional gauge theory $G$ with matter invariant under the center $Z(G)$ of $G$ is equivalent to a union of $G/Z(G)$ gauge theories with discrete theta angles \cite{Sharpe:2014tca}.
    
    \item If one restricts the instantons of four-dimensional Yang-Mills theory to have instanton number divisible by $k>1$ then the theory decomposes into $k$-many copies of regular four-dimensional Yang-Mills with differing theta angles  \cite{Tanizaki:2019rbk}.
\end{itemize}

Evidence that physics does indeed obey these decompositions can be seen through the physical quantities of the theories. For example, the torus partition function of 2d theories can be recast as the sum of partition functions of other theories as predicted by the decomposition formulae \cite{Hellerman:2006zs}. From this perspective, a possible decomposition is signalled by the behaviour of the partition function and the fact that it can be reorganised to correctly reproduce the partition functions of other theories. However, it may be difficult to see this splitting initially. In this paper we will focus on orbifold $\sigma$-models and show how the McKay quiver \cite{McKay1} of the orbifold group can be used as a simple indicator for decomposition.

McKay quivers have long been used in physics \cite{Lawrence:1998ja,Hanany:1998sd,Feng:2000af,Feng:2000mw,Aspinwall:2000xs} to understand string compactifications through the McKay correspondence \cite{itoreid,MR740077,bridge2001}. Despite this extensive study, little progress has been made into understanding the properties of disconnected McKay quivers from both a physical and algebraic point of view. Recently, a paper 
\cite{browne_2021} has shed some light onto the properties of disconnected McKay quivers and exhibited that there exist examples of disconnected components which do not arise as regular McKay quivers of finite groups. While these do not correspond to standard McKay quivers, we will see that these components can be understood as projective quivers \cite{Aspinwall:2000xs,Feng:2000af,Feng:2000mw} which arise in the study of orbifolds with discrete torsion \cite{Vafa:1986wx,Vafa:1994rv,Douglas:1996sw,Douglas:1998xa,Sharpe:1999pv,Sharpe:1999xw,Sharpe:2000ki}. In motivating this equivalence between QFT decomposition and McKay quiver decomposition, we thereby claim that disconnected McKay quivers can be understood precisely by the formulae of \cite{Hellerman:2006zs,Robbins:2020msp}, thus giving each component a definitive geometric meaning.

Physically, one can expect the decomposition of McKay quivers to correspond to the decomposition of $\sigma$-models with $\cN=2$ supersymmetry by remarking that the boundary conditions of these theories describe supersymmetric cycles in the target space geometry \cite{Govindarajan:2001kr,Lindstrom:2002jb,Lindstrom:2002mc,Hori:2000ck}. That is, they describe branes satisfying a BPS bound. In particular, using the standard brane-sheaf dictionary \cite{Aspinwall:2004jr, Aspinwall:2009isa}, one can identify this with the derived (bounded) category of coherent sheaves $D^b(\rm{Coh}(\cM))$ of the target orbifold $\cM$. The statement of decomposition is then equivalent to a statement about coproduct decompositions of derived categories of coherent sheaves on quotient stacks. Since McKay quivers are expected to capture the category of branes \cite{Douglas:1996sw,Hanany:1998sd, Closset:2019juk}, one expects that the McKay quiver reflects this decomposition appropriately.

In order to provide mathematical evidence for this equivalence, we will give a description of the components of the disconnected quivers defined by central extensions by using a purely group and representation theoretic argument\footnote{We will see that the natural language for understanding these quivers is given by group extensions and group cohomology. We recommend the following books for those unfamiliar with these topics \cite{karpilovsky1985projective, gruenberg1970cohomological, brown1982cohomology,beyl1982group, rotman2008introduction}.}. This generalises the work of \cite{Feng:2000af,Feng:2000mw}, where Schur extensions were used to easily calculate projective quivers to include more general central extensions. We will see that these quivers will be inline with the description of `banded' gerbes from \cite{Hellerman:2006zs}.

The paper is organised as follows. Section 2 will review some of the major results of the decomposition of orbifold theories. The first part of section 3 will serve as a brief review of linear and projective McKay quivers and their properties. The remainder of section 3 explains the connection between group extensions and disconnected McKay quivers. Finally, in section 4 we give several examples of quiver decomposition which showcase the equivalence of orbifold QFT decomposition and the decomposition of McKay quivers.

\section{Orbifold decomposition}
\subsection{The decomposition formulae}
Our main theories of interest are two dimensional nonlinear $\sigma$-models with target space $[X/G]$ for some manifold $X$ and discrete group $G\leq\rm{Isom}(X)$. We will denote such a theory by $\QFT([X/G])$. Implicit in this definition is the action of $G$ which, in principle, could have a trivially acting normal subgroup $N\trianglelefteq G$. In this case, $[X/G]$ should not really be thought of as a standard orbifold but instead a generalisation called a quotient stack \cite{Sharpe:2001bs}. Na\" ively, one may think that this should be equivalent to working on an orbifold $[X/(G/N)]$, but it turns out that physics detects the difference between using an orbifold $[X/(G/N)]$ and a stack $[X/G]$.

Intrinsically, the difference is captured by the symmetries of the two theories. By working on the quotient stack $[X/G]$ it is explicit that there is a one-form symmetry $\rm{B}Z(N)$, while the effective orbifold theory lacks such a symmetry. The presence of this one-form symmetry then signals that the theory can decompose into a theory on multiple disjoint spaces. This decomposition is captured by the formulae of \cite{Hellerman:2006zs,Robbins:2020msp}. 

Consider a theory $\QFT([X/G])$ with a trivially acting normal subgroup $N\trianglelefteq G$. This gives a short exact sequence
\begin{gather}\label{eq:exact}
    1\rightarrow N \rightarrow G \rightarrow G/N \rightarrow 1.
\end{gather}
Let $\hat{N}$ denote the set of isomorphism classes of irreducible representations (irreps) of $N$. Then there is a natural action of $G/N$ on $\hat{N}$ given by taking a lift $l\in G$ of $k\in G/N$ and defining $k\cdot \rho$ to be the representation $g\mapsto \rho(l g l^{-1})$ where $\rho\in \hat{N}$. The decomposition formula of \cite{Hellerman:2006zs} then states that
\begin{gather}\label{eq:decomp-orb}
    \QFT([X/G]) = \QFT\left(\ \bigsqcup_{\hat{\omega}} \left[\frac{X\times \hat{N}}{(G/N)}\right]_{\hat{\omega}}\right),
\end{gather}
where $\hat{\omega}$ encodes the discrete torsion of each theory (described explicitly in \cite{Hellerman:2006zs}). This simplifies tremendously if $N$ is a central subgroup of $G$. In this case the stabilizers of the action of $G/N$ are isomorphic and we get
\begin{gather}\label{eq:cent}
    \QFT([X/G]) = \QFT\left(\ \bigsqcup_{\rho\in \hat{N}}  \left[\frac{X}{(G/N)}\right]_{\hat{\omega}(\rho)}\right).
\end{gather}
Now the discrete torsion phase is determined as the image of the extension class $\omega\in\HH^2(G/N,N)$ defined by the central extension in \cref{eq:exact} under each $\rho\in\hat{N}$.

It is worth noting that this whole story can be generalised to the case where $[X/G]$ admits both a trivially acting subgroup and discrete torsion \cite{Robbins:2020msp}. As such, these formulae give a method of determining the disjoint orbifolds that occur whenever one works with a quotient stack $[X/G]$.

\subsection{A classic example}\label{sec:d4}
Should decomposition occur then the physical quantities derived from the theory must reflect the decomposition in question. In order to illustrate this, we recall a classic example considered in \cite{Pantev:2005rh,Hellerman:2006zs} to show how decomposition is reflected in the torus partition function of the theory.

Consider the space $[X/D_4]$ where the $\bbZ_2$ center of $D_4$ acts trivially. Following \cite{Hellerman:2006zs}, we denote the elements of $D_4$ by
\begin{gather}
D_4 = \{ 1, z, a, b, az, bz, ab, ba \}
\end{gather}
where $a^2 = b^4 = 1$, $b^2 = z$ and $ba=abz$. In this presentation the center of $D_4$ is simply $\bbZ_2\cong \{1, z\}$. According to \cref{eq:cent}, this theory decomposes into two pieces: one with discrete torsion and one without. Explicitly, we have
\begin{gather}
    \QFT([X/D_4]) = \QFT\big( [X/(\bbZ_2 \times
\bbZ_2)] \sqcup [X/(\bbZ_2 \times \bbZ_2)]_{\rm{d.t.}} \big),
\end{gather}
where $[X/(\bbZ_2 \times \bbZ_2)]_{\rm{d.t.}}$ denotes the $\bbZ_2\times\bbZ_2$ orbifold theory with discrete torsion. To see this physically, we can compute the one-loop torus partition function of the orbifold as
\begin{gather}
\cZ (D_4) = \frac{1}{|D_4|} \sum_{gh=hg} Z_{g,h}.
\end{gather}
As the center acts trivially, this reduces to
\begin{gather}\label{eq:part}
\cZ(D_4)  \: = \: \frac{1}{2}\left( \sum_{g \in \{1, a, b, ab\}} Z_{g,1} + \sum_{g \in \{1, a\}} Z_{g, a} + \sum_{g \in \{1, b\}} Z_{g, b} + \sum_{g \in \{1, ab \}} Z_{g, ab} \right).
\end{gather}

Now we need to check if this one-loop partition function is the same as the sum of the one-loop partition functions of the $\bbZ_2 \times \bbZ_2$ orbifold theory with and without discrete torsion. We present $\bbZ_2\times\bbZ_2$ as
\begin{gather}
    \bbZ_2\times\bbZ_2 = \langle a,b : a^2=b^2=1,\, bab=a\rangle.
\end{gather}
The possible discrete torsion of $\bbZ_2\times\bbZ_2$ is classified by $\HH^2(\bbZ_2\times\bbZ_2,\bbCt)=\bbZ_2$. In the sector with non-trivial discrete torsion, the torsional phases are given by \cite{Vafa:1994rv}
\begin{gather}
    \epsilon(a^m b^n,a^{m'}b^{n'})  = (-1)^{mn'-nm'}.
\end{gather}
Turning on discrete torsion results in multiplying each contribution of the torus partition function by the appropriate phase.

First we calculate the one-loop partition function with no discrete torsion. Explicitly, we have
\begin{gather}
\cZ(\bbZ_2 \times
\bbZ_2) =  \frac{1}{|\bbZ_2 \times
\bbZ_2|} \sum_{ g \in \{1, a, b, ab\}} \big(Z_{g,1} + Z_{g,a} + Z_{g,b} + Z_{g, ab}\big). 
\end{gather}
In the sector with discrete torsion, not all terms have positive coefficient. In particular, we get
\begin{gather}
 \cZ(\bbZ_2 \times
\bbZ_2)_{\epsilon}  =  \frac{1}{|\bbZ_2 \times
\bbZ_2|} \sum_{ g,h \in \bbZ_2 \times
\bbZ_2} \epsilon(g,h) Z_{g,h} = P_\epsilon - N_\epsilon,
\end{gather}
where the positive contribution is given by
\begin{gather}
    P_\epsilon = \frac{1}{4}\left( \sum_{ g \in \{1, a, b, ab\}} Z_{g,1} + \sum_{ g \in \{1, a\}} Z_{g,a} + \sum_{ g \in \{1, b\}} Z_{g,b} + \sum_{ g \in \{1, ab\}} Z_{g, ab}\right)
\end{gather}
and the negative contribution is
\begin{gather}
    N_\epsilon = \frac{1}{4}\left(\sum_{ g \in \{b, ab\}}  Z_{g,a} + \sum_{ g \in \{a, ab\}} Z_{g,b} + \sum_{ g \in \{a, b\}} Z_{g, ab} \right).
\end{gather}
If we add the one-loop partition functions of the $\bbZ_2 \times
\bbZ_2$ orbifold with and without discrete torsion, we recover the one-loop partition function of the $D_4$ orbifold with trivially acting center found in \cref{eq:part}

Another method of probing decomposition besides the partition function analysis is given by inspecting the cohomology of the theory. Consider $[X/D_4]$ with the same trivially acting $\bbZ_2$ center. In particular we take $X = T^6$, which leads to the Hodge diamond
\begin{displaymath}
\begin{array}{ccccccc}
 & & & 2 & & & \\
 & & 0 & & 0 & & \\
 & 0 & & 54 & & 0 & \\
2 & & 54 & & 54 & & 2 \\
 & 0 & & 54 & & 0 & \\
 & & 0 & & 0 & & \\
 & & & 2 & & & \end{array}
\end{displaymath}
At face value this appears to violate cluster decomposition due to the multiple dimension zero operators \cite{Hellerman:2006zs}, but we can look at this from  the perspective of the effective orbifold group $\bbZ_2\times\bbZ_2$ instead. The Hodge diamonds for the $\bbZ_2\times\bbZ_2$ orbifold of $T^6$ are given in \cite{Vafa:1994rv}. Explicitly, the Hodge diamond for the orbifold theory with no discrete torsion is given by
\begin{displaymath}
\begin{array}{ccccccc}
 & & & 1 & & & \\
 & & 0 & & 0 & & \\
 & 0 & & 51 & & 0 & \\
1 & & 3 & & 3 & & 1 \\
 & 0 & & 51 & & 0 & \\
 & & 0 & & 0 & & \\
 & & & 1 & & & \end{array}
\end{displaymath}
Similarly, the same calculation for the torsion sector gives
\begin{displaymath}
\begin{array}{ccccccc}
 & & & 1 & & & \\
 & & 0 & & 0 & & \\
 & 0 & & 3 & & 0 & \\
1 & & 51 & & 51 & & 1 \\
 & 0 & & 3 & & 0 & \\
 & & 0 & & 0 & & \\
 & & & 1 & & & \end{array}
\end{displaymath}
This cohomological data gives information about the massless spectrum of the theories and as such should respect any decomposition the theory respects. Indeed, adding the two $\bbZ_2\times\bbZ_2$ Hodge diamonds together gives us precisely the Hodge diamond for the $D_4$ theory with a trivially acting center. 

These simple examples provide compelling evidence that the physics of these theories do indeed obey the decomposition outlined by \cref{eq:decomp-orb}. The remarkable fact that the physics in some way knows about the trivially acting parts of the orbifold group is captured mathematically by the difference between quotient stacks and effective orbifolds. As remarked in \cite{ruan_2007}, while the orbit spaces of $[X/G]$ and $[X/G_{\rm{eff}}]$ are identical, the structure of the two spaces are intrinsically different as the trivially acting part of $G$ will appear in any chart of $[X/G]$ and not at all in $[X/G_{\rm{eff}}]$.

In this paper we will primarily focus on orbifolds of $\mathbb{C}^{2}$ and $\mathbb{C}^{3}$ in the special case in which the the orbifold group is a finite subgroup of $\SU(2)$ and $\SU(3)$ respectively. These are orbifolds of great interest in theoretical physics since they describe singular, non-compact, Calabi-Yau manifolds. These orbifolds describe the world volume theory of probe branes, which turn out to give superconformal field theories. This allows us to compare many of the quiver diagrams we obtain in \cref{sec:examples} with the preexisting literature on orbifold compactification \cite{Lawrence:1998ja, Hanany:1998sd,Feng:2000af, Feng:2000mw}.

\section{Ineffective orbifolds and McKay quivers}
Given that the decomposition formulae of \cite{Hellerman:2006zs,Robbins:2020msp,Sharpe:2022ene} depend on representations of a trivially acting subgroup, there is a natural question: do the McKay quivers of the ineffective orbifold group detect the decomposition of the QFT? In the language of representation theory, a trivially acting subgroup corresponds to choosing the action of the orbifold group to be unfaithful on $\bbC^n$. It is known that this causes the corresponding McKay quiver to be disconnected \cite{McKay1,ringel,browne_2021}, but we wish to be more explicit in the description of the disconnected components. 

Let us briefly review the basics of McKay quivers and their projective counterparts before giving a more detailed description of the disconnected quiver components.

\subsection{Linear McKay quivers}
Given a representation of a finite group, the McKay quiver summarises the tensor product decomposition of said representation when tensored against any irreducible representation of the group. This leads to the following definition \cite{McKay1}.
\begin{defn}
    Given a finite group $G$ with irreps  $\{\rho_i\}_{i=1}^n$, the McKay quiver relative to a (possibly reducible) representation $\cR$ is the quiver $\cQ(G,\cR)$ with adjacency matrix $A=(a_{ij})$ defined by
    \begin{gather}\label{eq:quiv}
        \cR \otimes \rho_i = \bigoplus_{j=1}^n  a_{ij}\rho_j.
    \end{gather}
\end{defn}

\noindent Using the orthogonality of characters, we can invert \cref{eq:quiv} to obtain the following expression for the adjacency matrix
\begin{gather}
    a_{ij} = \frac{1}{|G|}\sum_{g\in G}\chi_\cR (g) \chi_{i}(g) \overline{\chi_j(g)},
\end{gather}
where $\chi_i(g)=\rm{Tr}\,\rho_i (g)$.

Despite being constructed in a purely representation theoretic manner, McKay quivers have a rich interpretation in terms of geometry and K-theory \cite{MR740077,itoreid,ROAN1996489,bridge2001}. The McKay correspondence, and its generalisations, tell us that the McKay quiver of a group $G$ encodes the geometry of an orbifold $[X/G]$. This information was leveraged to understand the matter content of theories descending from string theory compactified on $[X/G]$ \cite{Lawrence:1998ja,Douglas:1996sw,Hanany:1998sd}. More recently, for finite subgroups of $\SU(3)$, these quivers have been interpreted as 5d BPS quivers for theories coming from M-theory compactified on $[\bbC^3/G]$ \cite{Closset:2019juk,DelZotto:2022fnw}.

To end this section, let us recall the following useful property of McKay quivers. Given an irreducible representation of any dimension we can tensor multiply it with a one dimensional irrep to obtain another irrep of the same dimension. Since the irreps of $G$ label the vertices of the quiver, we find that this generates an automorphism of the quiver. This group, isomorphic to $\ab{G}=G/G^{(1)}$ where $G^{(1)}=[G,G]$, is sometimes called the quantum symmetry of the corresponding orbifold theory.
\subsection{Projective McKay quivers}
McKay quivers describe the matter content of orbifold theories with no discrete torsion. However, discrete torsion can be incorporated by additionally looking at projective representations of the orbifold group \cite{Douglas:1998xa,Aspinwall:2000xs}. 

A projective representation of a finite group $G$ is a homomorphism $P:G\rightarrow \PGL{V}$. Any such $P$ must then satisfy
\begin{gather}
    P(x)P(y) = \alpha(x,y)P(xy),
\end{gather}
for some function $\alpha:G\times G\rightarrow \bbCt$. Furthermore, associativity and $P(1_G)=1_V$ constrain $\alpha$ to satisfy
\begin{gather}
    \alpha(x,y)\alpha(xy,z)=\alpha(x,yz)\alpha(y,z),\quad \alpha(x,1_G)=1=\alpha(1_G,x).
\end{gather}
These conditions simply state that $\alpha$ is a 2-cocycle. Multiplying any 2-cocycle by a 2-coboundary 
\begin{gather}
    (\delta\beta)(x,y)=\frac{\beta(x)\beta(y)}{\beta(xy)},
\end{gather}
leads to equivalent projective representations. From this we see that the possible projective representations are classified by the Schur multiplier $\HH^2(G,\bbCt)$. We call a projective representation labelled by a 2-cocycle $\alpha$ an $\alpha$-representation.

Just as linear representations can be thought of as modules over the group algebra $\bbC G$, projective representations with 2-cocycle $\alpha$ can be thought of as modules over the $\alpha$-twisted group algebra $\bbC^\alpha G$. We describe $\bbC^\alpha  G$ by specifying a basis given by $\{e_x: x\in G\}$ and supplementing it with the distributive product
\begin{gather}
    e_x\cdot e_y = \alpha(x,y)e_{xy}.
\end{gather}
This was used in \cite{Aspinwall:2000xs} in order to define a projective McKay quiver in terms of module data and projective characters. In particular, let $P_i$ be the simple $\bbC^\alpha G$ modules corresponding to irreducible $\alpha$-representations. Then given a general $\bbC^\alpha G$-module $\cP$, we can decompose it into
\begin{gather}
    \cP = \bigoplus_i V_i \otimes P_i,
\end{gather}
where $V_i$ are vector spaces. If $\cR$ is a linear represenation and $\cP$ an $\alpha$-representation, then it is clear that $\cR\otimes\cP$ is also an $\alpha$-representation. Applying the above decomposition to $\cR\otimes P_i$ then gives the analogue of \cref{eq:quiv}
\begin{gather}
    \cR\otimes P_i = \bigoplus_j V_{ij}\otimes P_j.
\end{gather}
The corresponding quiver adjacency matrix is then $a_{ij}=\dim V_{ij}$. This quiver, which we denote by $\cQ_\alpha(G,\cR)$, then encodes the same information as a regular McKay quiver but now with discrete torsion specified by $\alpha\in\HH^2(G,\bbCt)$.

A disadvantage of this approach is that the values of $\alpha(x,y)$ are needed explicitly in order to compute the quiver. It is therefore advantageous to lift the projective representations to linear representations of another group, as in \cite{Feng:2000af, Feng:2000mw}. Doing so, one can then use regular character theory in order to compute torsion quivers.
\subsection{Extensions and decompositions}\label{sec:extanddecomp}
Having covered both the torsional and non-torsional cases we are now ready to describe the decomposition of McKay quivers. To address the question of when a McKay quiver is disconnected, we quote the following theorem (see \cite{browne_2021} for a proof and discussion).

\begin{theorem}
    The number of connected components of a McKay quiver relative to $\rho:G\rightarrow \GL{V}$ is given by the number of $G$-conjugacy classes contained in $\ker\rho$.
\end{theorem}

\noindent From this we see that we obtain disconnected quivers precisely when the orbifold action is unfaithful. Furthermore, \cite{ringel} also gives us a description of the component connected to the trivial representation. It is simply the McKay quiver of $G/\ker\rho$ relative to the quotient representation of the orbifold action. Describing the remaining components is more subtle, but it is helpful to view the unfaithful representation $\rho$ as defining an extension of groups
\begin{gather}\label{eq:ext}
    1\rightarrow N \xrightarrow{\iota} G \rightarrow F \rightarrow 1,
\end{gather}
where $N=\ker\rho$ and $F=G/N$. These extensions can fall into three classes:
\begin{enumerate}
    \item If we have that $\iota(N)$ does not lie in the center of $G$, then we call the extension non-central. In particular, this is the case whenever $N$ is non-abelian.
    \item If $\iota(N)\leq Z(G)$, then the extension is called central. Such extensions are classified by $\HH^2(F,N)$, so \cref{eq:ext} yields a (possibly trivial) 2-cocycle of this group.
    \item Finally, if $\iota(N)\leq Z(G)\cap G^{(1)}$, where $G^{(1)}=[G,G]$ is the derived subgroup of $G$, then the extension is called a stem extension.
\end{enumerate}
In the latter two cases we can describe the quiver decomposition rather explicitly\footnote{The non-central extension case is much more involved from a group cohomology point of view. We hope to return to this case in a future work.}. 

We recall that there is a set of maximal stem extension which correspond to taking $N$ to be the Schur muliplier $\HH^2(F,\bbCt)$ \cite{gruenberg1970cohomological}. The resulting extension groups are called Schur covering groups or representation groups for $F$. Schur covering groups have the remarkable property that every projective representation of $F$ lifts to a linear representation of $G$ \cite{karpilovsky1985projective}. Therefore, if $G$ is a Schur cover of $F$, then the McKay quiver relative to $\rho$ is simply the collection of all torsional quivers of $G/N$ with each component corresponding to a different value of $\alpha\in\HH^2(F,\bbCt)$. Indeed, this fact was exploited in \cite{Feng:2000mw, Feng:2000af} to generate quivers with discrete torsion easily without appealing to projective characters and twisted group algebras. In this form we have that
\begin{gather}\label{eq:schurcov}
    \cQ(G, \rho) = \bigsqcup_{\alpha\in\HH^2(F,\bbCt)} \cQ_\alpha(F, \tilde{\rho}),
\end{gather}
where $\tilde{\rho}$ denotes the quotient representation of $\rho$ on $F$.

We can refine this further. Suppose that \cref{eq:ext} is a stem extension but $G$ is not a Schur cover of $F$. Then it is known that this extension is the homomorphic image of some Schur cover $G^*$ \cite{gruenberg1970cohomological,rotman2008introduction}. This means that we have the commutative diagram
\begin{gather}
\begin{aligned}
    \xymatrix{1 \ar[r] & \HH^2(F,\bbCt) \ar[r]^{\iota^*}\ar[d]^{j} & G^* \ar[r]^{\xi^*}\ar[d]^k & F \ar[r] \ar[d] & 1 \\
              1 \ar[r] & A \ar[r]^{\iota} & G \ar[r]^{\xi} & F \ar[r] & 1}
\end{aligned}
\end{gather}
Since $G$ is the homomorphic image of some Schur cover $G^*$ we note that, by Schur's lemma, the image of the irreducible representations of $G^*$ under $k$ are either the trivial representation or the irreducible representations of $G$. By then lifting the quotient representation of $F$ to $G^*$, we see that the McKay quiver relative to $\rho$ is a full subquiver of the McKay quiver relative to the lift on $G^*$. Furthermore, since $A$ is the image of $\HH^2(F,\bbCt)$ we expect that the McKay quiver relative to $\rho$ will only contain lifts of projective representations corresponding to a subgroup of $\HH^2(F,\bbCt)$ isomorphic to $A$\footnote{Here we are noting that $A\cong \HH^2(F,\bbCt)/\ker j$ by the first isomorphism theorem. Since both $A$ and $\HH^2(F,\bbCt)$ are finite abelian groups, there exists a subgroup of $\HH^2(F,\bbCt)$ isomorphic to $A$.}. An example of this will be shown in \cref{ex:stem}.

It turns out that the central, but not stem, extension case is largely determined by the stem extension case. In order to motivate this, we define the notion of isoclinicism of groups.

\begin{defn}
    Two finite groups $G_1$ and $G_2$ are said to be isoclinic if there exist isomorphisms $\eta:G_1/Z(G_1)\rightarrow G_2/Z(G_2)$ and $\nu:G_1^{(1)}\rightarrow G_2^{(1)}$ commuting with the commutator map.
\end{defn}

\noindent Isoclinic groups, in a sense, encode similar representations. The linear representations of such groups are lifts of projective representations of the inner automorphism group $\rm{Inn}(G)\cong G/Z(G)$ where each representation is lifted with an appropriate multiplicity. Stated more concretely, let $m_{i}^{d}$ denote the number of degree $d$ representations of $G_i$. If $G_1$ and $G_2$ are isoclinic, then \cite{beyl1982group}
\begin{gather}
    \frac{m_1^d}{m_2^d}=\frac{|G_1|}{|G_2|},
\end{gather}
for any $d\in\bbN$. In fact, all Schur covering groups of a given group are isoclinic \cite{karpilovsky1985projective}. 

Since isoclinic groups encode the same projective representations up to multiplicity, it should be clear that the quivers for isoclinic groups relative to lifts of the same linear representation are related. For the case of Schur covers, it was remarked in \cite{Feng:2000af} that differing the Schur cover will lead to the same quiver. We also posit that the same holds for isoclinic non-Schur covers of the same order and present a derivation of this in \cref{sec:proof}.

The following theorem due to P. Hall (see \cite{gruenberg1970cohomological} for a modern proof) allows us to use isoclinicism to reduce the central extension case to a stem case. 
\begin{theorem}\label{thm:stem}
     Let $E$ be a central extension of $G$ by $A$. Then $E$ is isoclinic to a stem extension $F$ of $G$ by some abelian group $B$.
\end{theorem}

\noindent Note that isoclinism requires that the orders of the derived subgroups of $E$ and $F$ coincide. Using the identity
\begin{gather}
    \left[\frac{E}{A},\frac{E}{A}\right]=\frac{A[E,E]}{A},
\end{gather}
we arive at the condition
\begin{gather}
    |B|=|A\cap [G,G]|.
\end{gather}
Since the RHS divides $|A|$, we must have $|B|$ divides $|A|$ as well. As such, given a central extension $1\rightarrow A\rightarrow E \rightarrow G\rightarrow 1$ isoclinic to a stem extension $1\rightarrow B\rightarrow F \rightarrow G\rightarrow 1$, we must have that $|F|$ divides $|E|$. Since $k=|E|/|F|$ is an integer, we have that $E$ is of the same order as the isoclinic group $F\times \bbZ_k$. The quiver for $F\times \bbZ_k$ is easily derived and leads to the following proposition.
\begin{proposition}\label{prop:dis}
    Let $E$ be a central extension of $G$ isoclinic to a stem extension $F$ of $G$ and $\cR$ a faithful representation of $G$ with lifts $\cR_E$ and $\cR_F$. Writing $|E|/|F|$ as $k\in\bbZ$, we have 
    \begin{gather}
        \cQ(E,\cR_E) = \bigsqcup_{i=1}^{k} \cQ(F,\cR_F).
    \end{gather}
    In other words, the quiver for the extension $E$ is simply a number of copies of the quiver for the stem extension $F$.
\end{proposition}

We end this section with a brief comparison with \cite{Hellerman:2006zs,Pantev:2005rh,Sharpe:2022ene}. We have seen that when \cref{eq:ext} is a stem extension, the group $N=\ker\rho$ determines a subgroup of $\HH^2(G/N,\bbCt)$ from which we obtain a subset of all torsional quivers. The central extension case then follows by moving to an isoclinic stem extension and appropriately replicating the quiver. In each case we see that we only obtain a disjoint union of orbifolds of the type $[X/(G/N)]$ with discrete torsion, and multiplicity, determined by a subgroup of $\HH^2(G/N,\bbCt)$. This is in agreement with the `banded' examples of \cite{Hellerman:2006zs}.

\medskip
\noindent {\bf Note on vertex labelling.} In the quiver diagrams we present we have chose to keep the vertices unlabelled as opposed to furnishing them with the dimension of the corresponding representation. In the central and stem extension case, all dimensions would match what one would expect from the decomposition of QFTs. However, in the non-central case there is necessarily a mismatch in the dimensions due to the action of $\ab{G}$ on the quiver. In particular, the action of $\ab{G}$ causes any two components containing one dimensional representations to be equal \cite{browne_2021}. Despite this, the overall shape of the graph is in total agreement.

\section{Examples}\label{sec:examples}
In this section we will give several examples which show how McKay quivers capture decomposition in several differing settings. All major group data for the groups we cover is presented in \cref{app:gr}.
\subsection{\texorpdfstring{$D_4$}{D4} revisited}

Let's start by analysing the example from \cref{sec:d4} in our framework. We wish to realise the central extension
\begin{gather}\label{eq:d4}
    1\rightarrow \bbZ_2 \rightarrow D_4 \rightarrow \bbZ_2\times\bbZ_2\rightarrow 1,
\end{gather}
by choosing a representation of $D_4$ which trivialises exactly the $\bbZ_2$ center. This is easy to see from the character table of $D_4$:
\begin{gather*}
    \begin{tabular}{c|ccccc}
    \hline \hline 
    & $C_1^{(1)}$ & $C_2^{(2)}$ & $C_3^{(2)}$ & $C_4^{(2)}$ & $C_5^{(1)}$\\  \hline
$\rho_1$ &1 & 1 & 1 & 1 & 1 \\
$\rho_2$ &1 & $-1$ & $-1$ & 1 & 1 \\
$\rho_3$ &1 & $-1$ & 1 & $-1$ & 1 \\
$\rho_4$ &1 & 1 & $-1$ & $-1$ & 1 \\
$\rho_5$ &$2$ & 0 & 0 & 0 & $-2$ \\ \hline\hline
\end{tabular}
\end{gather*}
Here $C_i^{(j)}$ denotes a conjugacy class of size $j$. The center corresponds exactly to all the size one conjugacy classes, so we have that $Z(D_4)=C_1^{(1)}\cup C_5^{(1)}$. Recall that a representation $\rho$ trivialises a conjugacy class $C_i^{(j)}$ if its character satisfies
\begin{gather}
    \chi_\rho(g) = \chi_\rho(1_G),
\end{gather}
for a representative $g$ of $C_i^{(j)}$. This immediately tells us that $\rho_2\oplus\rho_3\oplus\rho_4$ trivialises exactly $Z(D_4)$, so we take this as the representation to compute the McKay quiver relative to. The resulting quiver is shown in \cref{fig:z2z2}.

\begin{figure}
    \centering
    \includegraphics[scale=0.35]{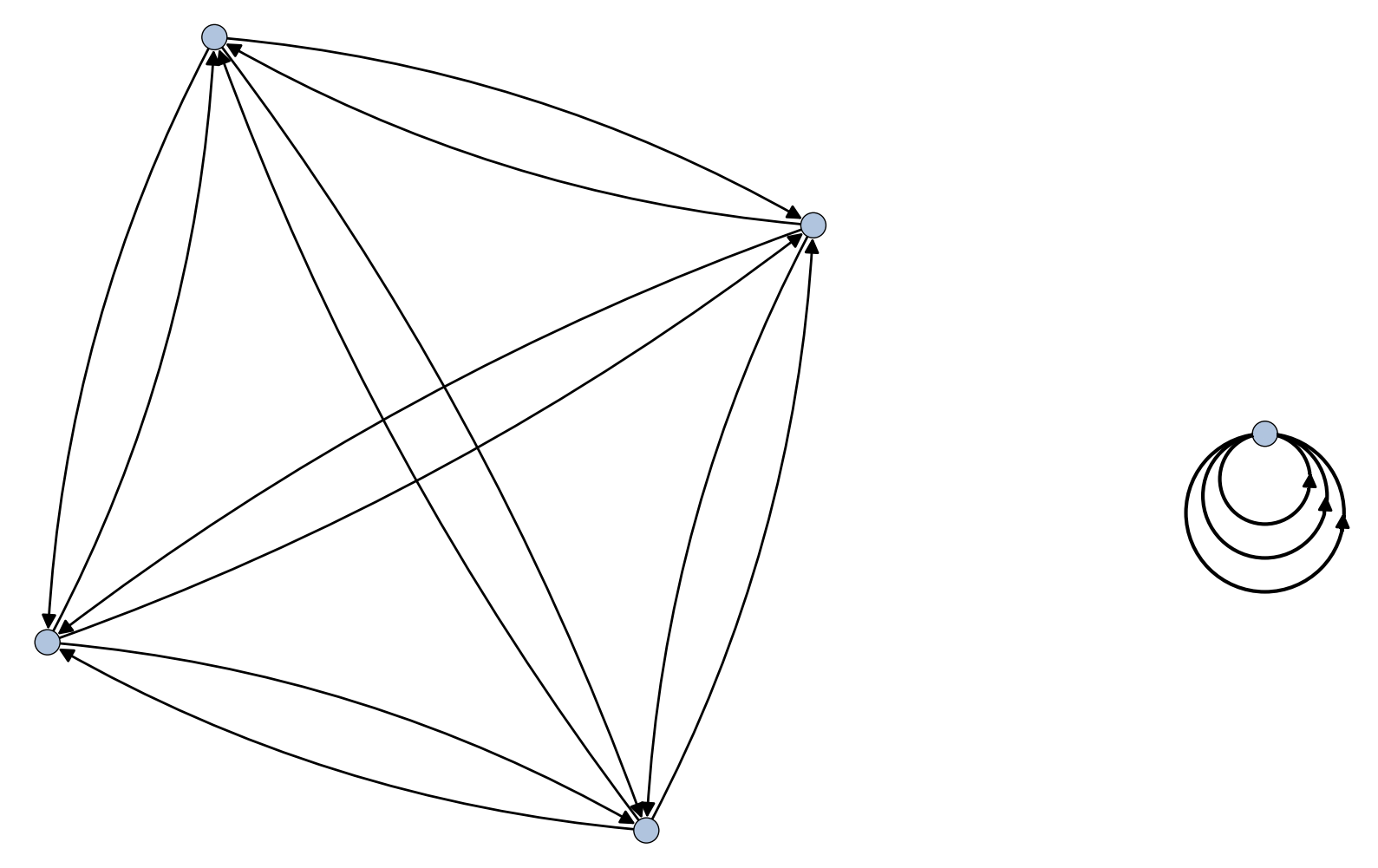}
    \caption{The McKay quiver for $D_4$ which trivialises the center. On the left is the standard McKay quiver for $\bbZ_2\times\bbZ_2$ while the quiver on the right is the projective McKay quiver corresponding to the non-trivial element of $\HH^2(\bbZ_2\times\bbZ_2,\bbCt)=\bbZ_2$.}
    \label{fig:z2z2}
\end{figure}

Comparing with \cite{Aspinwall:2000xs}, we recognise these as precisely the regular McKay quiver for $\bbZ_2\times\bbZ_2$ and the only possible projective McKay quiver for $\bbZ_2\times\bbZ_2$. The fact that we obtain all the torsional quivers arises from the fact that $D_4$ is a Schur cover of $\bbZ_2\times\bbZ_2$. Indeed, it is easy to see that $\HH^2(\bbZ_2\times\bbZ_2,\bbCt)=\bbZ_2$ and that \cref{eq:d4} is a stem extension. We therefore see that this is consistent with the QFT decomposition
\begin{gather}
    \QFT([X/D_4]) = \QFT\big( [X/(\bbZ_2 \times
\bbZ_2)] \sqcup [X/(\bbZ_2 \times \bbZ_2)]_{\rm{d.t.}} \big).
\end{gather}

\subsection{More involved examples}\label{sec:nontrivialex}
The classification of finite subgroups of $\SU(3)$ includes two infinite series $\Delta(3n^2)$ and $\Delta(6n^2)$ \cite{Blichfeldt1905,yau1993gorenstein}. When $n=2$ we have the following incidental isomorhpisms
\begin{gather}
    \Delta(3\cdot 2^2)\cong \Alt(4), \quad \Delta(6\cdot 2^2)\cong \Sym(4)
\end{gather}
Since the Schur covering groups of these have been well studied \cite{1992projective}, we shall use both of these to generate examples of decomposition. 
    
\begin{figure}
    \centering
    \includegraphics[scale=0.41]{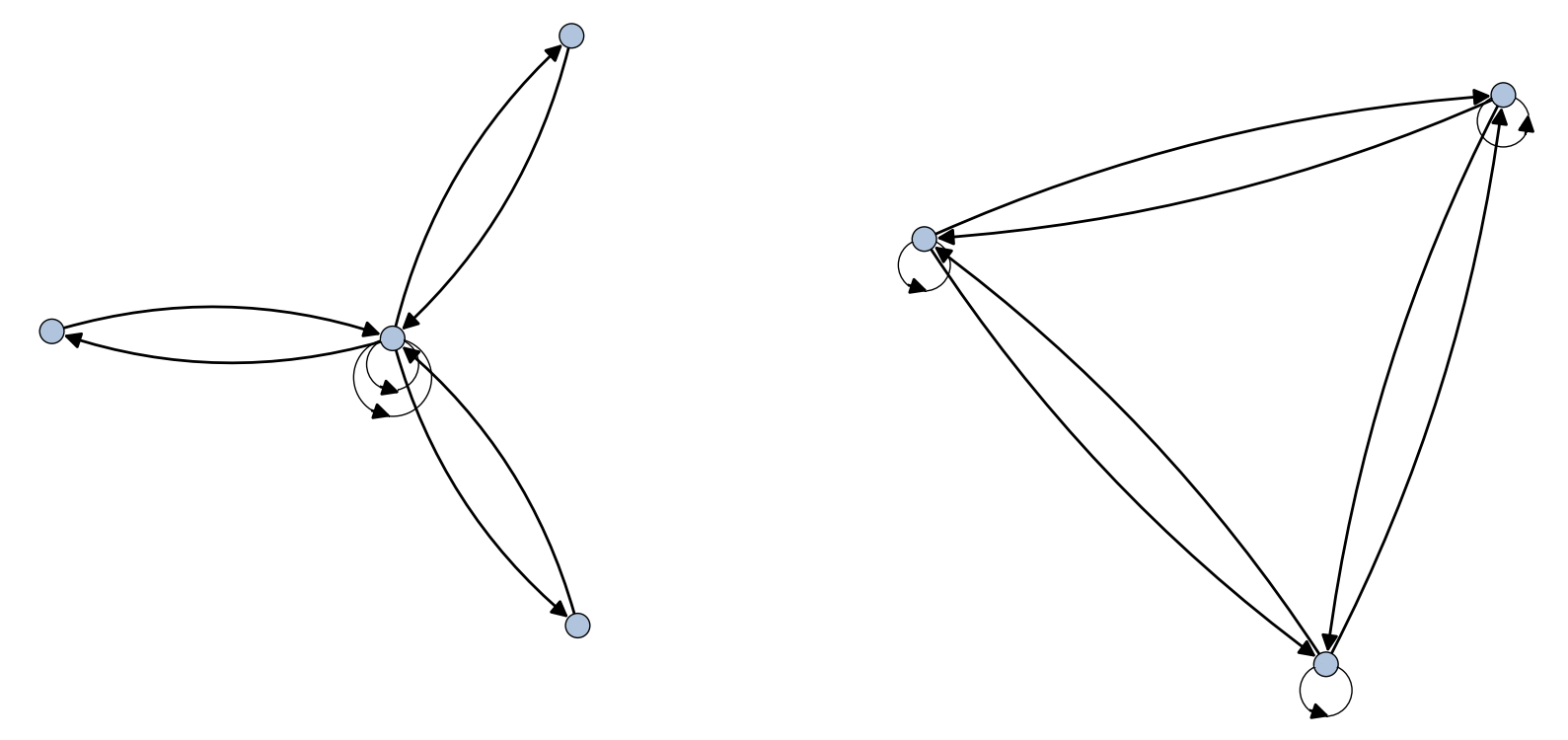}
    \caption{The unfaithful McKay quiver for the binary tetrahedral group. The component on the left realises the standard McKay quiver for $\Alt(4)$, while the righthand quiver gives the torsional quiver for $\Alt(4)$.}
    \label{fig:tetalt}
\end{figure}    
    
First of all, we note that the Schur multiplier of $\Alt(4)$ is simply $\bbZ_2$ and a Schur covering group can be given by the stem extension
\begin{gather}
    1\rightarrow \bbZ_2 \rightarrow \rm{SL}_2(\bbF_3)\rightarrow \Alt(4)\rightarrow 1.
\end{gather}
Interestingly, the Schur cover $\rm{SL}_2(\bbF_3)$ is isomorphic to the binary tetrahedral group $\cT$-- meaning it can be found as a subgroup of $\SU(3)$. As such, by taking the orbifold action to be the three dimensional irreducible representation that trivialises the $\bbZ_2$ center, we get the decomposition
\begin{gather}
    \QFT([\bbC^3/\cT]) = \QFT([\bbC^3/\Alt(4)]\sqcup[\bbC^3/\Alt(4)]_{\rm{d.t.}}).
\end{gather}
We show the quiver for this decomposition in \cref{fig:tetalt}.

A similar computation can also be done for the case of $\Sym(4)$. However, instead of simply repeating the above computation, we use this orbifold to explore a more nuanced decomposition predicted by \cite{Robbins:2020msp} which can be correctly reproduced in our framework. Consider $\Sym(4)$ with a non-central trivially acting $\bbZ_2\times\bbZ_2$ subgroup. Then decomposition predicts that
\begin{gather}\label{eq:s4s3}
    \QFT([X/\Sym(4)]_{\rm{d.t.}}) = \QFT([X/\Sym(3)]).
\end{gather}
To deal with an orbifold with both discrete torsion and a non-central non-trivially acting subgroup, we need to take a two step approach. 

\begin{figure}
    \centering
    \includegraphics[scale=0.42]{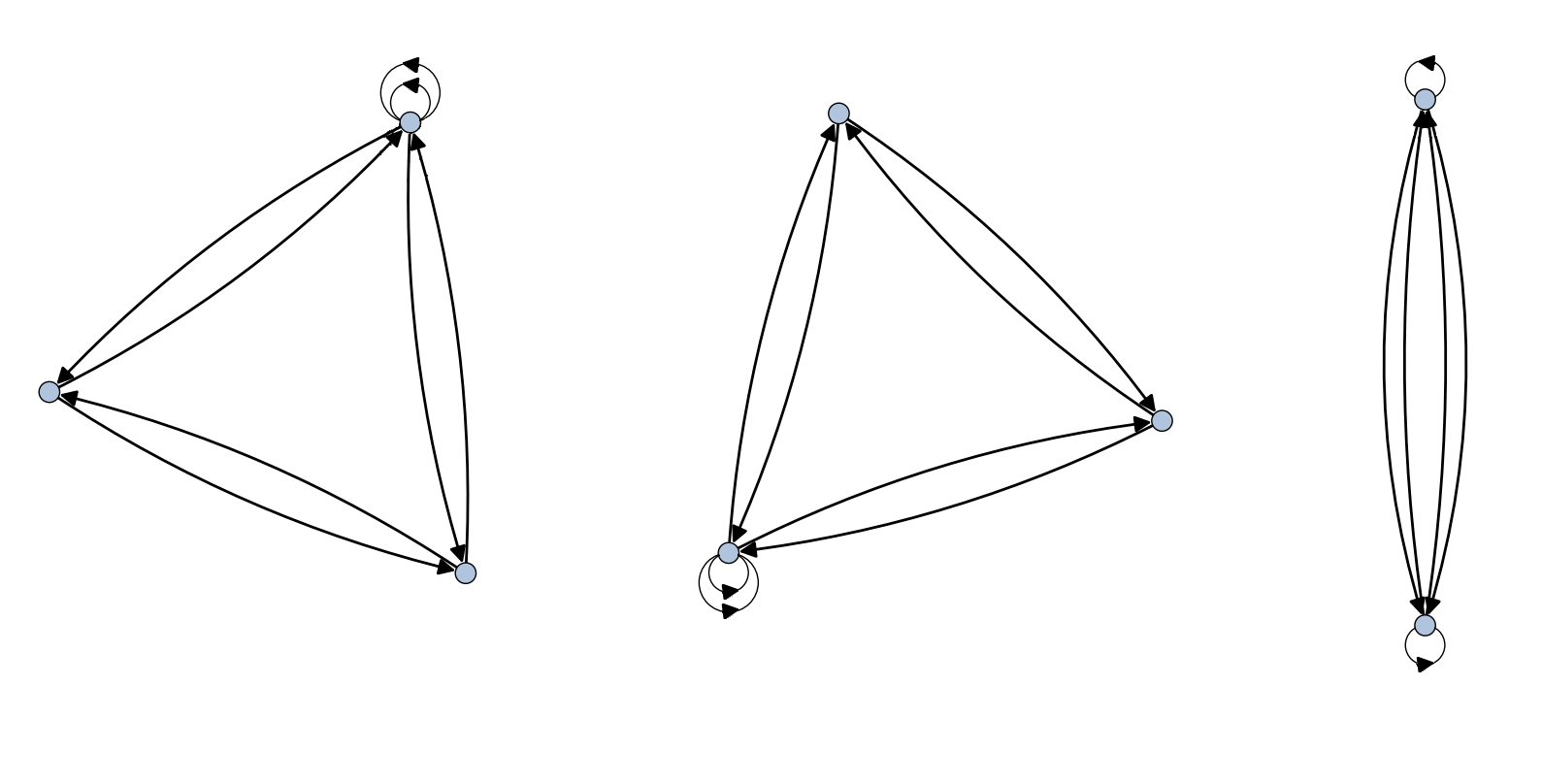}\vspace*{-0.25in}
    \caption{The McKay quiver for $\rm{GL}_2(\bbF_3)$ that realises the lift of the representation that trivialises the $\bbZ_2\times\bbZ_2$ subgroup of $\Sym(4)$.}
    \label{fig:s4schur}
\end{figure}

\medskip
\noindent {\bf Step 1.} First we compute the quiver of $[\bbC^3/\Sym(4)]$ with respect to the action that trivialises the $\bbZ_2\times\bbZ_2$ subgroup. This action corresponds to viewing $\Sym(4)$ as the extension
\begin{gather}
    1\rightarrow \bbZ_2 \times \bbZ_2 \rightarrow \Sym(4) \rightarrow \Sym(3) \rightarrow 1.
\end{gather}
The McKay quiver of this set up then gives rise to the two rightmost components in \cref{fig:s4schur}. This is itself an additional example of decomposition, but we postpone the details until after the next step.

\medskip
\noindent {\bf Step 2.} Now to find the quiver of $[\bbC^3/\Sym(4)]$ with discrete torsion we take the lift of the representation used above to a Schur cover. A possible Schur covering group is given by the stem extension
\begin{gather}
    1\rightarrow \bbZ_2 \rightarrow \rm{GL}_2(\bbF_3) \rightarrow \Sym(4) \rightarrow 1.
\end{gather}
The lift of the representation now has a kernel which spans three conjugacy classes of $\rm{GL}_2(\bbF_3)$, corresponding to the lift of the $\bbZ_2 \times \bbZ_2$ subgroup of $\Sym(4)$. This gives us the three component McKay quiver in \cref{fig:s4schur}. Since $\rm{GL}_2(\bbF_3)$ is a Schur cover, this corresponds to the decomposition
\begin{align}
    \QFT([\bbC^3/\rm{GL}_2(\bbF_3)]) = \QFT([\bbC^3/\Sym(4)] \sqcup [\bbC^3/\Sym(4)]_{\rm{d.t.}}),
\end{align}
with a trivially acting $\bbZ_2 \times \bbZ_2\trianglelefteq \Sym(4)$.

Therefore, the correct quiver for $[\bbC^3/\Sym(4)]$ with discrete torsion and a trivially acting $\bbZ_2\times \bbZ_2$ is the quiver given by taking away the components that arise from step 1 in \cref{fig:s4schur}. This leaves us with just the leftmost component. By inspecting the character table of $\Sym(3)$ it is not hard to obtain such a quiver, thus confirming \cref{eq:s4s3}. Furthermore, we can identify the rightmost component as a $\bbZ_2$ quiver, giving us the additional decomposition
\begin{gather}
    \QFT([\bbC^3/\Sym(4)]) = \QFT([\bbC^3/\Sym(3)]\sqcup [\bbC^3/\bbZ_2]),
\end{gather}
as predicted by \cref{eq:decomp-orb}.

\subsection{Stem but not Schur}\label{ex:stem}
Now consider the group $\Delta(108)=\Delta(3\cdot 6^2)$. This group has a $\bbZ_3$ center whose action we wish to trivialise on $\bbC^3$. To do so, we choose one of the two three dimensional irreps that trivialise exactly the center. This leads to the extension
\begin{gather}
        1\rightarrow \bbZ_3 \rightarrow \Delta(108) \rightarrow \bbZ_3 \times \rm{Alt}(4) \rightarrow 1.
\end{gather}
Furthermore, it is easy to see that this is a stem extension. The quiver we obtain is given in \cref{fig:delta108}. While this is a valid stem extension, $\Delta(108)$ is not a Schur cover of $\bbZ_3\times\Alt(4)$. Indeed, we have that
\begin{gather}\label{eq:z3altext}
    \HH^2(\bbZ_3,\bbCt)=1,\quad \HH^2(\Alt(4),\bbCt)=\bbZ_2, \quad \ab{\Alt(4)}=\bbZ_3.
\end{gather}
Therefore, by \cref{thm:prod}, we have that
\begin{gather}
    \HH^2(\bbZ_3\times\Alt(4),\bbCt) = \bbZ_2\times\bbZ_3\cong \bbZ_6.
\end{gather}
As such, we expect to only see a subset of possible torsional quivers arising from this. To confirm that the additional components in \cref{fig:delta108} are indeed torsional variants of the $\bbZ_3\times\Alt(4)$ quiver, we must look at the Schur extension
\begin{gather}
    1\rightarrow \bbZ_6 \rightarrow Q_8 \rtimes \rm{He}_3 \rightarrow \bbZ_3\times\Alt(4) \rightarrow 1,
\end{gather}
where $Q_8\cong \rm{Dic}_2$ is the group of quaternions and $\rm{He}_3$ is the mod-3 Heisenberg group. Indeed, lifting a faithful representation of $\bbZ_3\times \Alt(4)$ to the covering group gives the 6 component quiver in \cref{fig:qhe}, three components of which are exactly those in \cref{fig:delta108}, thus confirming that we are only seeing the $\bbZ_3\leq \HH^2(\bbZ_3\times\Alt(4),\bbCt)$ subgroup of the torsional quivers when considering $\Delta(108)$.

\begin{figure}
    \centering
    \includegraphics[scale=0.33]{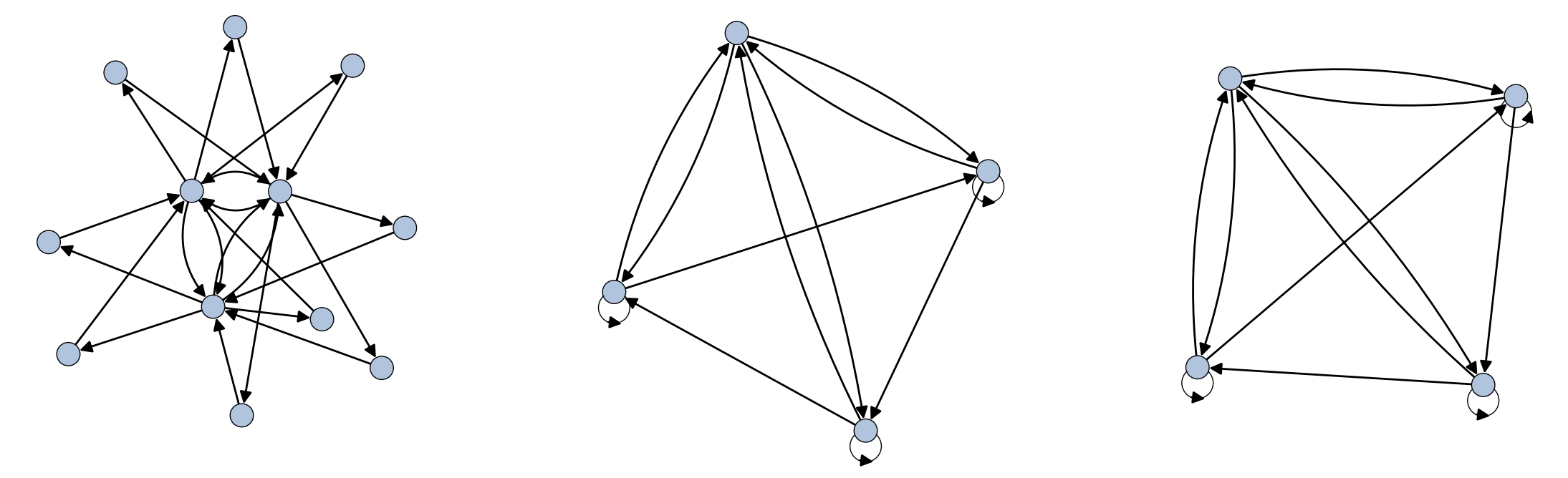}
    \caption{The quiver for $\Delta(108)$ that trivialises the $\bbZ_3$ center. We recognise the McKay quiver for $\bbZ_3\times\Alt(4)$ together with two additional components. These can then be identified as torsional quivers for $\bbZ_3\times\Alt(4)$ by considering the Schur extension \cref{eq:z3altext}.}
    \label{fig:delta108}
\end{figure}
\begin{figure}
    \hspace*{-0.33in}
    \includegraphics[scale=0.44]{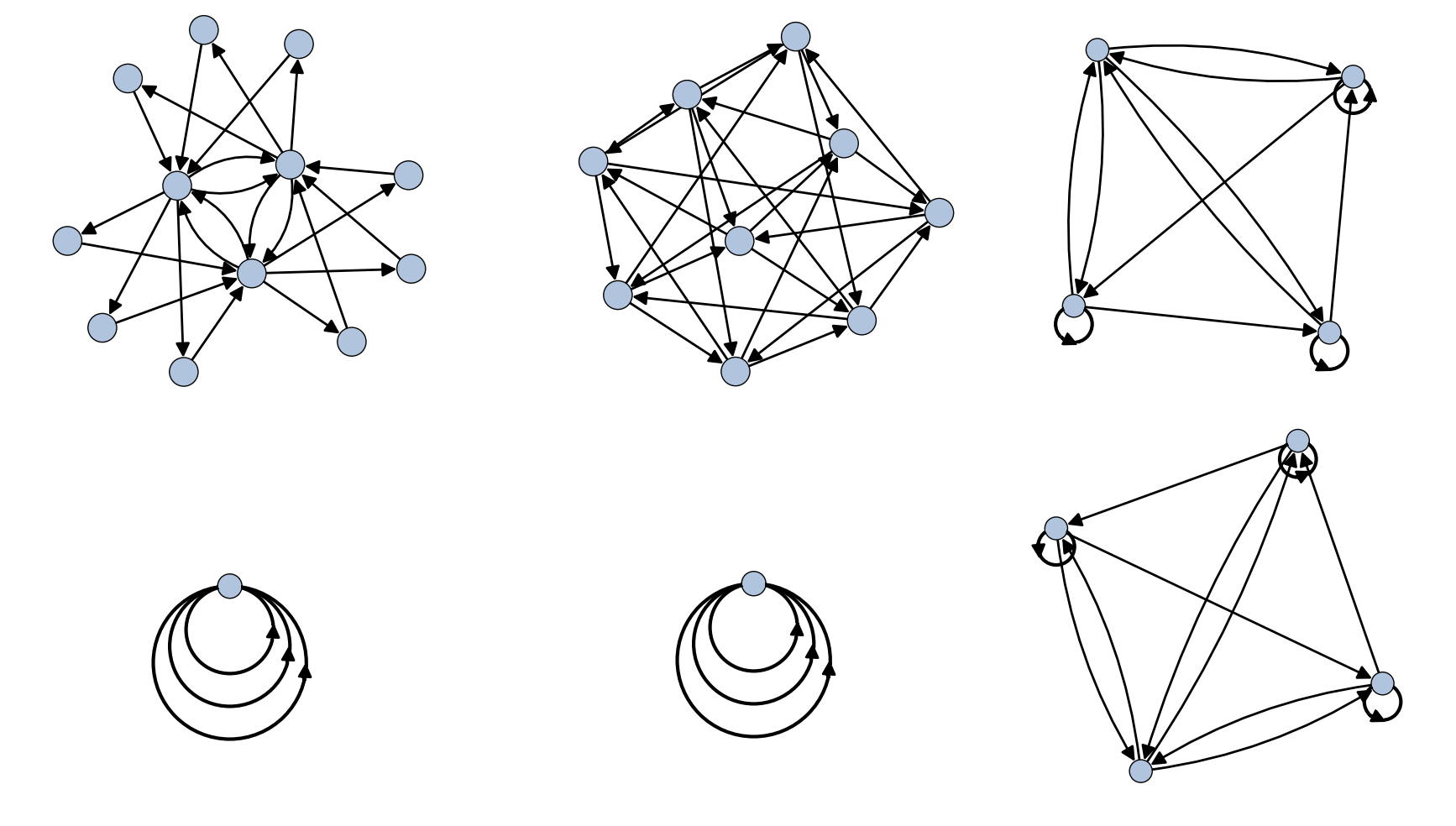}
    \caption{The disconnected McKay quiver for $Q_8\rtimes \rm{He}_3$. As this is a Schur cover for $\bbZ_3\times\Alt(4)$ we obtain all torsional quivers corresponding to discrete torsion labelled by $\alpha\in\HH^2(\bbZ_3\times\Alt(4),\bbCt)$.}
    \label{fig:qhe}
\end{figure}

As an illustration of \cref{prop:dis}, consider the central extension
\begin{gather}
    1\rightarrow \bbZ_6\rightarrow \bbZ_2\times \bbZ_3^2\cdot \Alt(4)\rightarrow \bbZ_3\times \Alt(4)\rightarrow 1,
\end{gather}
where $\bbZ_3^2\cdot\Alt(4)$ is the non-split extension of $\Alt(4)$ by $\bbZ_3^2$ acting via $\Alt(4)/\bbZ_2^2=\bbZ_3$. This is a non-stem central extension isoclinic to $\Delta(108)$, signalling the corresponding quiver should be exactly a double copy of the one in \cref{fig:delta108}. Indeed, plotting the quiver for $\bbZ_2\times \bbZ_3^2\cdot \Alt(4)$ confirms this.

\subsection{A Schur trivial example}
So far we have seen examples where the effective orbifold group accomodates enough discrete torsion so that the extension can see either all of the projective representations or a subset them. It is natural to ask what happens if we take a central extension of a group that cannot have any discrete torsion.

Consider the regular dihedral group of 10 elements $D_5$. We first note that we have
\begin{gather}
    \HH^2(D_5, \bbCt)=1, \quad \HH^2(D_5,\bbZ_2)=\bbZ_2.
\end{gather}
As such, we see that while the orbifold theory $[X/D_5]$ does not accommodate discrete torsion, there are two equivalence classes of $\bbZ_2$ central extensions of it. In particular, we have the two extensions
\begin{gather}
    \begin{aligned}
        \xymatrix{& & D_{10} \ar[dr] \\
        1\ar[r] &\bbZ_2 \ar[ur]\ar[dr] & & D_5\ar[r] & 1\\
        & & \rm{Dic}_5 \ar[ur]}
    \end{aligned}
\end{gather}
Note that $D_{10}\cong \bbZ_2 \times D_5$, so this corresponds to the trivial central extension while $\rm{Dic}_5$ provides a non-trivial extension. Since $D_5$ has trivial Schur multiplier, the only stem extension of it is the trivial extension by $1$. \Cref{thm:stem} then implies that any central extension of $D_5$ is isoclinic to $D_5$. As such, any central extension of $D_5$ by an abelian group $A$ simply gives rise to $|A|$ many copies of the $D_5$ quiver. For the case of $D_{10}$ and $\rm{Dic}_5$, this is drawn in \cref{fig:dic}.

\begin{figure}
    \centering
    \includegraphics[scale=0.4]{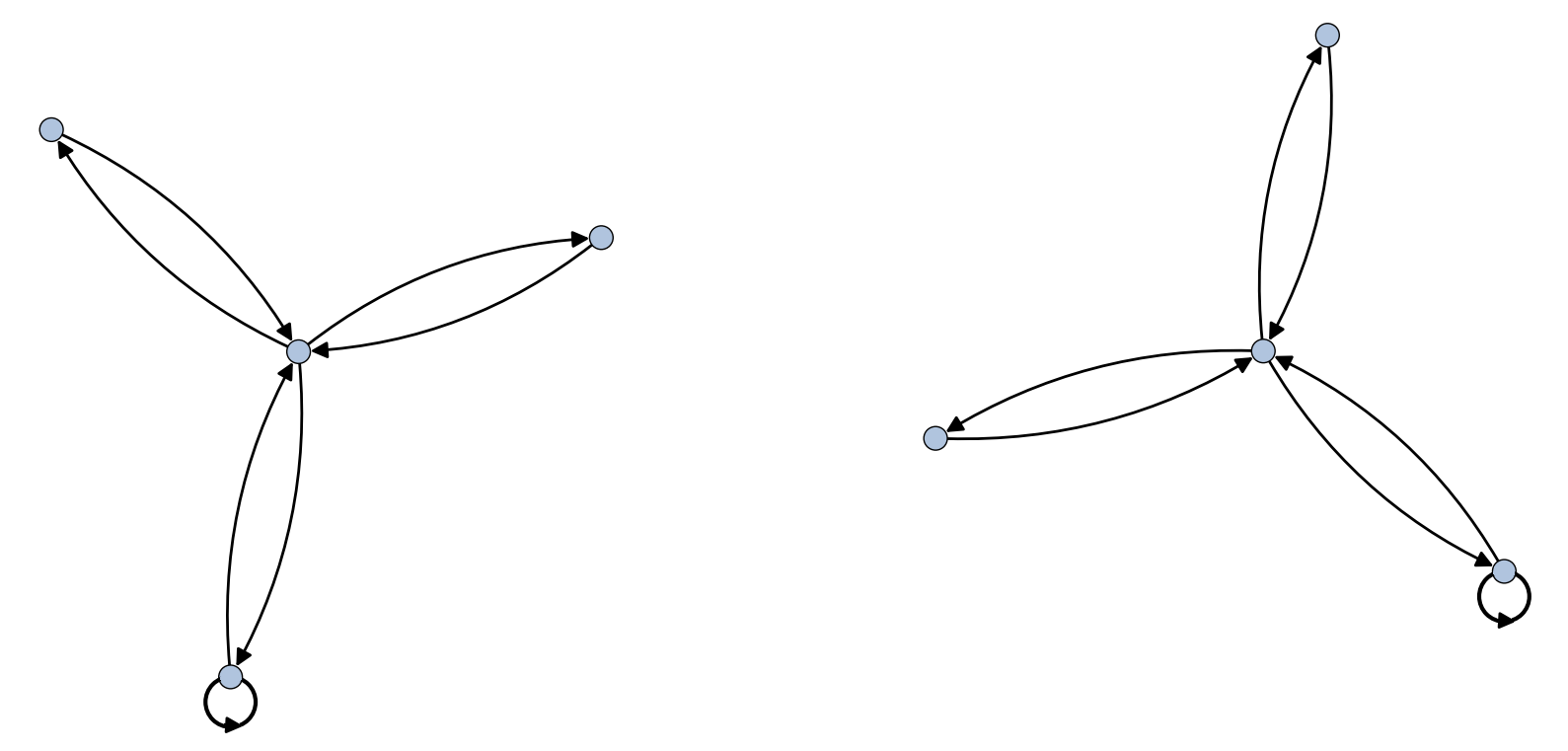}
    \caption{The McKay quiver for $\rm{Dic}_5$ which trivialises a central $\bbZ_2$ subgroup. Here we take the representation that lives in $\rm{U}(2)$ to illustrate that the decomposition is not unique to $\SU(n)$ representations.}
    \label{fig:dic}
\end{figure}

By using the above argument we can take any central extension of $D_5$ by some abelian group $A$ and obtain a quiver which is just a number of copies of the $D_5$ McKay quiver. This then provides many examples of isoclinic groups of the same order giving the same quiver whenever $\HH^2(D_5,A)$ is non-trivial.  

\section{Conclusion}
In this paper we have motivated the equivalence of the decomposition of QFTs and the decomposition of McKay quivers. In doing so, we derived formulae describing the disconnected components of such quivers for central extensions using purely group and representation theoretic means which agree with the decomposition of QFT formulae of \cite{Hellerman:2006zs, Sharpe:2022ene}.  

Indeed, we showed that orbifolds with a trivially acting central subgroup produce a McKay quiver capturing the decomposition features of the theory. Moreover, in the particular case of central and stem extensions, the McKay quiver components are simply a subset of the quivers for the same orbifold with and without discrete torsion. Furthermore, we have seen that McKay quivers relative to lifts of linear representations are sensitive only to the isoclinism class of the central extension, a feature described in \cref{sec:extanddecomp}. Finally, these properties provide a method to study the effects of decomposition when a non-central subgroup of the orbifold group acts trivially, both for linear and projective representations as shown in \cref{sec:nontrivialex}.

There are several possible further directions of this work:
\begin{enumerate}
    \item The formulae we derived are only valid for central extensions. The non-central extension case is much more complicated and, as seen in the $\Sym(4)$ example, can result in decompositions into orbifolds with non-isomorphic orbifold groups. It is natural to ask if our analysis can be extended to these cases to obtain a formula similar to \cref{eq:decomp-orb} to describe the general case. Non-central, and hence non-abelian, extensions are no longer classified by $\HH^2(G,A)$ only, so a more sophisticated analysis is require in this case.
    \item McKay quivers have been widely used to study theories arising from compactification of M-theory or string theory on an orbifold. While our analysis has applied to two dimensional $\sigma$-models, it would be interesting to understand the decomposition of these quivers from the point of view of these compactifications. 
    \item In a similar vein, in the context of geometric engineering McKay quivers provide a tool to study the BPS states of certain 4d $\mathcal{N}=2$ SQFT through intepretting them as BPS quivers \cite{Closset:2019juk,DelZotto:2022fnw}. It therefore encodes properties of fractional branes, which in the context of algebraic geometry are described as coherent sheaves on a certain variety. It would be interesting to explore the link between geometry and decomposition in this context and use it to explore the properties of the related field theory.
    \item As seen in \cite{Sharpe:2021srf,Robbins:2021ibx,Robbins:2021xce,Komargodski:2020mxz, Lin:2022xod}, decomposition can be used to understand the properties of the disjoint theories. Since the McKay quiver is sensitive to the decomposition, it is possible that the quivers are also sensitive to other information from which we can deduce information about the individual theories. Understanding which properties the quiver is sensitive to could provide new algebraic techniques for understanding the associated QFTs.
    \item From a holography point of view, the quiver represents the theory living on probe D-branes at the tip of a Calabi-Yau singularity. Theories with and without discrete torsion where already studied in \cite{Douglas:1996sw,Douglas:1998xa}, but theories displaying decomposition are still to be explored. It would be interesting to understand if the decomposition can be captured by the gravitational dual of theories living on D-branes and if there are physical observables that relates the different universes.
\end{enumerate}

\section*{Acknowledgements}
We thank Michele Del Zotto, Darius Dramburg, Jonathan Heckman, Ethan Torres and Hao Zhang for many helpful conversations. We also thank Michele Del Zotto and Eric Sharpe for suggestions and remarks on a preliminary version of the paper. The work of SM is funded by the scholarship granted by ``Fondazione Angelo Della Riccia". The work of RM has recieved funding from the Simons Foundation Grant \#888984 (Simons Collaboration on Global Categorical Symmetries) and is in the context of the ERC project MEMO, which is funded by the European Research Council (ERC) under the European Union’s Horizon 2020 research and innovation programme (grant agreement No. 851931).
\appendix
\section{Some group and character theory}\label{app:gr}
\subsection{Group cohomology and computation}
In this appendix we recount the basics of group cohomology and gather some theorems used for computations in this paper. Most statements, with proofs, can be found in the classic text \cite{karpilovsky1985projective}\footnote{In this paper we are only interested in the case where the coefficient module is an abelian group. We refer the reader interested in the more general situation to \cite{brown1982cohomology}.}.

Let $G$ be a finite group, $A$ an abelian group and $f:G^{n}\rightarrow A$. We can define a set of differentials on such functions as
\begin{gather}
    (\delta_{n+1} f)(g_1,\ldots, g_{n+1}) = f(g_2,\ldots,g_{n+1})\left(\prod_{i=1}^{n} f(g_1,\ldots,g_{i-1},g_i g_{i+1},\ldots,g_{n+1})^{(-1)^i} \right) f(g_1,\ldots,g_n)^{(-1)^{n+1}}
\end{gather}
where we are writing the product structure on $A$ multiplicatively. It is tedious, but easy, to see that $(\delta_{n+1}\circ\delta_n)f=1_A$ for any $f$ and thus defines a chain complex
\begin{gather}
    C^0(G,A) \xrightarrow{\delta_1} C^1(G,A) \xrightarrow{\delta_2} C^2(G,A) \xrightarrow{\delta_3}C^3(G,A) \xrightarrow{\delta_4} \cdots
\end{gather}
where $C^n(G,A)$ denotes the set of functions $f:G^n\rightarrow A$. We define the $n^{\rm{th}}$-cohomology group to be
\begin{gather}
    \HH^n(G,A) = \rm{ker}\,\delta_{n+1}/\rm{im}\,\delta_n.
\end{gather}
Furthermore, the elements of $\rm{ker}\,\delta_{n+1}$ are called $n$-cocycles and the elements of $\rm{im}\,\delta_n$ are called $n$-coboundaries.

The description of $\HH^1(G,A)$ is particularly simple. It is clear that $\rm{im}\,\delta_1$ is trivial, while $\delta_2 f =1_A$ imposes
\begin{gather}
    f(x)f(xy)^{-1} f(y) = 1_A.
\end{gather}
In other words, $\ker \delta_2\cong \rm{Hom}(G,A)$ which we identify with $\HH^1(G,A)$. If $A=\bbCt$ then this is simply the abelianization of $G$.

In this paper we are primarily interested in the second cohomology group which is central to the theory of group extensions. In simple cases, the above definition is enough to determine $\HH^2(G,A)$. For more intricate groups, however, we can use several theorems to simplify calculations tremendously. We list some theorems which were particularly useful for our purposes. 

\begin{theorem}\label{thm:prod}
    Let $G_1$ and $G_2$ be arbitrary groups and $A$ an abelian group. Then 
    \begin{gather}
        \HH^2(G_1\times G_2,A)\cong \HH^2(G_1,A) \times \HH^2(G_2,A) \times \rm{Hom}(G_1\otimes G_2,A),
    \end{gather}
    where $G_1\otimes G_2$ is the group $\ab{G_1}\otimes_\bbZ \ab{G_2}$.
\end{theorem}

\begin{theorem}
    Let $Z$ be a central subgroup of a finite group $G$. Then we have
    \begin{gather}
        \HH^2(G,\bbCt) = \HH^2(G/Z,\bbCt)/(Z\cap G^{(1)}),
    \end{gather}
    where $G^{(1)}=[G,G]$ is the derived subgroup of $G$.
\end{theorem}

\begin{theorem}
     Let $G$ be a finite group and $\HH_2(G,\bbZ)$ be its second homology group with coefficients in $\bbZ$. There exists an isomorphism
    \begin{gather}
        \HH^2(G,\bbCt)\cong \HH_2(G,\bbZ).
    \end{gather}
\end{theorem}
\noindent {\it Note.} We can typically compute $\HH_2(G,\bbZ)$ in {\tt Sage} or {\tt GAP}. 

These theorems together with {\tt Sage} and {\tt GAP} are sufficient to find the second cohomology groups presented in \cref{tbl:grp} which are used throughout the paper.
\begin{table}[]
\centering
\begin{tabular}{cccc}
\hline
\hline\\[-1.1em]
$G$                    & $\HH^2(G,\bbCt)$ & $G^*$                    & {\tt GAP ID}          \\ \hline
$D_5$ & $1$ & $D_5$ & {\tt [10,1]}\\ 
$\rm{Dic}_5$ & $1$ & $\rm{Dic}_5$ & {\tt [20,1]} \\ 
$\bbZ_2\times\bbZ_2$ & $\bbZ_2$ & $D_4$ & {\tt [8,3]} \\
$\Alt(4)$              & $\bbZ_2$         & $\rm{SL}_2(\bbF_3)\cong \cT$ & {\tt [24,3]}         \\
$D_{10}$ & $\bbZ_2$& $\bbZ_5 \rtimes D_4$& {\tt [40,8]}\\
$\Sym(4)$              & $\bbZ_2$         & $\rm{GL}_2(\bbF_3)$       & {\tt [48,29]}        \\
$\bbZ_3\times \Alt(4)$ & $\bbZ_6$            & $Q_8\rtimes \rm{He}_3$   & {\tt [216,42]}      \\ \hline \hline
\end{tabular}\caption{Here we list the groups together with the Schur cover that we used within the paper. The final column gives the {\tt GAP ID} of $G^*$, so the interested reader can access them via the {\tt SmallGroups} library.}\label{tbl:grp}
\end{table}
\subsection{Proof of isoclinism invariance of McKay quivers}\label{sec:proof}

\begin{proposition}
    Let $G_1$ and $G_2$ be two isoclinic central extensions of some group $F$ by $A_1$ and $A_2$ respectively. Further assume that $G_1$ and $G_2$ have the same order. Then
    \begin{gather}
        \cQ(G_1, \cR_1) = \cQ(G_2, \cR_2),
    \end{gather}
    where $\cR_1$ and $\cR_2$ are lifts of a linear representation $\cR:F\rightarrow \GL{V}$.
\end{proposition}
\begin{proof}
    Let $\cP$ be a (possibly projective) representation of $F$. Then lifting $\cP$ to a central extension $G$ we see that $\cP(f)$ lifts to $\cP'(\mu(f)\cdot A)$ where $\mu:F\rightarrow G$ is a section. Since $A$ is central, we must have that $\cP'(A)$ is a non-zero multiple of the identity. Denote this as $\cP'(a_i)=\lambda_\cP(a_i)1_{\GL{V}}$. Note that if $\cP$ is actually linear, then $\lambda_\cP(a_i)\equiv 1$. The key observation here is that the $\lambda_\cP$ are simply irreducible representations of $A$. Furthermore, from $\lambda_\cP$ we can recover the 2-cocycle associated to $\cP$  by defining the function $\Phi:G\times G\rightarrow A$ as \cite{1992projective}
    \begin{gather}
        \mu(x)\cdot \mu(y) = \Phi(x,y)\mu(xy).
    \end{gather}
    Then $\alpha(x,y)=\lambda_\cP(\Phi(x,y))$ is the associated 2-cocycle to $\cP$. As such, two representations $\cP_1$ and $\cP_2$ with equal center representations $\lambda_{\cP_1}\equiv\lambda_{\cP_2}$ have the same associated 2-cocycle.
    
    The adjacency matrix for the McKay quiver of $G$ relative to a linear lift $\cR_G$ is given by
    \begin{gather}
        M_{ij} = \frac{1}{|G|}\sum_{g\in G} \chi_{\cR_G}(g) \chi_i(g)\overline{\chi_j(g)}.
    \end{gather}
    We can partition $G$ into $A$-cosets and rewrite the sum as
    \begin{gather}
        M_{ij} = \frac{1}{|G|}\sum_{\ell\in \mu(F)}\sum_{a\in A} \chi_{\cR_G}(\ell a) \chi_i(\ell a)\overline{\chi_j(\ell a)}.
    \end{gather}
    As mentioned above, since $A$ is central, we must have that $\cP'(a)$ is a scalar matrix and we can rewrite the characters as
    \begin{gather}\label{eq:split}
        M_{ij} =\left( \frac{1}{|A|}\sum_{a\in A} \lambda_i(a)\overline{\lambda_j(a)}\right)
        \left(\frac{|A|}{|G|}\sum_{\ell\in \mu(F)}\chi_{\cR_G}(\ell) \chi_i(\ell)\overline{\chi_j(\ell)}\right).
    \end{gather}
    The first term simply enforces the orthogonality of characters. As such, this is only non-zero when $\lambda_i$ and $\lambda_j$ are equivalent representations and therefore determine the same 2-cocycle. The second term is more subtle, but does not depend on $G$, only on the values of projective representations of $F$. Since isoclinic groups of the same order have representations which are in one-to-one correspondence with each other and are lifts of projective representations of $G_1/Z(G_1)\cong G_2/Z(G_2)$ (with multiplicity), we see that \cref{eq:split} gives the same result up to relabelling of the indices. 
\end{proof}

Note that \cref{eq:split} also reproduces the result that the component connected to the trivial representation is simply the quiver for $F$ relative to $\cR$. It also shows that two vertices can be connected by an edge if and only if their center representations are equivalent.

\subsection{Characters and normal subgroups}
One of the benefits of the method advocated in this paper is that very little information about the explicit representations is needed. In fact, most of the pertinent information is encoded in the character table of these groups, which can be readily calculated through a modern computational algebra system such as {\tt Sage} or {\tt GAP}. We will briefly review how to obtain such information given the character table of a finite group.

Given a representation $\rho:G\rightarrow\GL{V}$ it is obvious that $\ker \rho$ is a normal subgroup of $G$. This kernel can be described as
\begin{gather}
    \ker\rho = \{ g \in G : \chi_\rho(g)=\chi_\rho(1_G)\}.
\end{gather}
Since characters are class functions, this guarantees that only full conjugacy classes are included in the kernel, as expected. It is also clear from this description that intersections of kernels corresponds to the addition of characters. That is, we have
\begin{gather}
    \ker\rho_1 \cap \ker\rho_2 = \{g\in G: \chi_{\rho_1\oplus\rho_2}(g)=\chi_{\rho_1\oplus\rho_2}(1_G)\}.
\end{gather}
It is known that all normal subgroups of $G$ can be obtained by taking the intersections of kernels of irreducible representations \cite{isaacs1994}. 

As we know all normal subgroups can be found from the character table, it is natural to ask if how can one find the center and derived subgroup of $G$ easily from the table alone. For the latter we have the simple expression
\begin{gather}
    G^{(1)} = \bigcap_{\rm{dim}\,\rho=1} \ker\rho.
\end{gather}
In other words, $G^{(1)}$ can be found by taking the intersection of the kernels of one-dimensional representations. Finding the center is a bit more subtle. To do so, we define the sets
\begin{gather}
    \zeta(\rho) =\{g\in G: |\chi_\rho(g)|=\chi_\rho(1_G)\} =\{g\in G : \rho(g) \in \bbCt 1_{\GL{V}}\}.
\end{gather}
The center of $G$ is then given as
\begin{gather}
    Z(G) = \bigcap_{\rho}\, \zeta(\rho), 
\end{gather}
where the intersection is over irreps of $G$. In particular, a one dimensional representation has $\zeta(\rho)=G$ so it suffices to look at representations of degree $2$ and higher only.

To illustrate the above statements consider the group $\rm{Dic}_3$, the character table of which is given by:
\begin{center}
\begin{tabular}{c|cccccc}
\hline \hline
\, & $C_1^{(1)}$ & $C_2^{(1)}$ & $C_3^{(2)}$ & $C_4^{(2)}$ & $C_5^{(3)}$ & $C_6^{(3)}$ \\ \hline 
$\rho_1$ & $1$ & $1$ & 1 & 1 & 1 & 1 \\
$\rho_2$ & 1 & $-1$ & 1 & $-1$ & $\I$ & $-\I$ \\
$\rho_3$ & 1 & 1 & 1 & 1 & $-1$ & $-1$ \\
$\rho_4$ & 1 & $-1$ & 1 & $-1$ & $-\I$ & $\I$ \\
$\rho_5$ & $2$ & $2$ & $-1$ & $-1$ & 0 & 0 \\
$\rho_6$ & $2$ & $-2$ & $-1$ & 1 & 0 & 0 \\ \hline \hline
\end{tabular}
\end{center}
A simple example of a normal subgroup is $N_1=C_1^{(1)}\cup C_3^{(2)}$. This is a non-central subgroup of order three which is exactly trivialised under $\rho_2$. Instead looking at the intersection of the kernels of $\rho_1$ through $\rho_4$ we get that the derived subgroup $\rm{Dic}_3^{(1)}$ is also given by $N_1$. Finally, the center of $\rm{Dic}_3$ is obtained by looking at $\zeta(\rho_5)$ and $\zeta(\rho_6)$ from which we get that $Z(\rm{Dic}_3)=C_1^{(1)}\cup C_3^{(1)}\cong \bbZ_2$, in agreement with the fact that central elements belong to their own conjugacy class.

\phantomsection
\addcontentsline{toc}{section}{References}
\bibliography{reps.bib}{}
\end{document}